\documentclass[11pt,a4paper]{article}
\usepackage[margin=1.1in]{geometry} %
\usepackage[utf8]{inputenc}
\usepackage{mathtools} %
\usepackage{amssymb}

\usepackage{microtype}
\usepackage{enumerate}
\usepackage{accents} %
\usepackage{graphicx}
\usepackage[table,dvipsnames]{xcolor}

\usepackage[colorlinks, pagebackref=true]{hyperref}
\hypersetup{
  linkcolor=[rgb]{0.3,0.3,0.6},
  citecolor=[rgb]{0.2, 0.6, 0.2},
  urlcolor=[rgb]{0.6, 0.2, 0.2}
}
\hypersetup{breaklinks=true}

\allowdisplaybreaks[1]

\usepackage{amsthm}
\usepackage{thmtools, thm-restate}
\usepackage[capitalize, nameinlink]{cleveref}

\declaretheorem[name=Theorem, parent=section]{theorem}
\declaretheorem[name=Corollary, sibling=theorem]{corollary}

\declaretheorem[name=Lemma, sibling=theorem]{lemma}

\declaretheorem[name=Definition, sibling=theorem, style=definition]{definition}
\declaretheorem[name=Remark, sibling=theorem, style=definition]{remark}

\crefname{theorem}{Theorem}{Theorems}
\crefname{lemma}{Lemma}{Lemmas}
\crefname{proposition}{Proposition}{Propositions}
\crefname{corollary}{Corollary}{Corollaries}
\crefname{conjecture}{Conjecture}{Conjectures}
\crefname{observation}{Observation}{Observations}
\crefname{fact}{Fact}{Facts}
\crefname{definition}{Definition}{Definitions}
\crefname{remark}{Remark}{Remarks}
\crefname{problem}{Problem}{Problems}
\crefname{example}{Example}{Examples}
\crefname{question}{Question}{Questions}
\crefname{claim}{Claim}{Claims}
\crefname{section}{Section}{Sections}
\crefname{subsection}{Section}{Sections}
\crefname{subsubsection}{Section}{Sections}

\DeclareMathOperator{\subrank}{Q}

\DeclareMathAccent{\wtilde}{\mathord}{largesymbols}{"65}
\DeclareMathOperator{\asympsubrank}{\underaccent{\wtilde}{Q}}

\DeclareMathOperator{\asymprank}{\underaccent{\wtilde}{R}}

\DeclareMathOperator{\tensorrank}{R}
\newcommand{\FF}{\mathbb{F}}

\newcommand{\RR}{\mathbb{R}}
\newcommand{\CC}{\mathbb{C}}
\newcommand{\NN}{\mathbb{N}}
\newcommand{\ZZ}{\mathbb{Z}}

\newcommand{\eps}{\varepsilon}
\newcommand{\field}{\FF}
\newcommand{\kron}{\boxtimes}

\def\<#1>{\left\langle\ignorespaces#1\unskip\right\rangle}

\newcommand{\asympleq}{\lesssim}

\usepackage{scalerel}
\DeclareMathOperator*{\bigkron}{\raisebox{-0.75ex}{\scaleobj{1.5}{\kron}}}

\newcommand{\func}{F}
\newcommand{\asympfunc}{\underaccent{\wtilde}{\func}}

\DeclareMathOperator{\GL}{GL}

\DeclareMathOperator{\linspan}{span}

\DeclareMathOperator{\Sym}{Sym}

\DeclarePairedDelimiterX\braket[2]{\langle}{\rangle}{#1\,\delimsize\vert\,\mathopen{}#2}

\newcommand{\lincombiparts}{p}

\title{Asymptotic tensor rank is characterized by polynomials}

\author{Matthias Christandl\thanks{Department of Mathematical Sciences, University of Copenhagen, Universitetsparken 5, 2100 Copenhagen, Denmark} \and Koen Hoeberechts\thanks{Korteweg-de Vries Institute for Mathematics, University of Amsterdam, Science Park 105--107, 1098 XG Amsterdam, Netherlands} \and Harold Nieuwboer\footnotemark[1] \and P\'eter Vrana\thanks{Department of Algebra and Geometry, Institute of Mathematics, Budapest University of Technology and Economics, M\H uegyetem rkp.\ 3., H-1111 Budapest, Hungary} \and Jeroen~Zuiddam\footnotemark[2]}
\date{}
\begin{document}
\maketitle

\begin{abstract}
Asymptotic tensor rank, originally developed to characterize the complexity of matrix multiplication, is a parameter that plays a fundamental role in problems in mathematics, computer science and quantum information. This parameter is notoriously difficult to determine; indeed, determining its value for the $2\times 2$ matrix multiplication tensor would determine the matrix multiplication exponent, a long-standing open problem. 

Strassen's asymptotic rank conjecture, on the other hand, makes the bold statement that asymptotic tensor rank equals the largest dimension of the tensor and is thus as easy to compute as matrix rank. Recent works have proved strong consequences of Strassen's asymptotic rank conjecture in computational complexity theory. Despite tremendous interest, much is still unknown about the structural and computational properties of asymptotic rank; for instance whether it is computable.

We prove that asymptotic tensor rank is ``computable from above'', that is, for any real number $r$ there is an algorithm that determines, given a tensor $T$, if the asymptotic tensor rank of $T$ is at most $r$. The algorithm has a simple structure; it consists of evaluating a finite list of polynomials on the tensor. Indeed, we prove that the sublevel sets of asymptotic rank are Zariski-closed (just like matrix rank). While we do not exhibit these polynomials explicitly, their mere existence has strong implications on the structure of asymptotic rank. 

As one such implication, we find that the values that asymptotic tensor rank takes on all tensors is a well-ordered set. In other words, any non-increasing sequence of asymptotic ranks stabilizes (``discreteness from above''). In particular, for the matrix multiplication exponent (which is the log of an asymptotic rank) there is no sequence of exponents of bilinear maps that approximates it arbitrarily closely from above without being eventually constant. In other words, any such upper bound on the matrix multiplication exponent that is close enough, will ``snap'' to it. Previously such discreteness results were only known for finite fields or for other tensor parameters (e.g., asymptotic slice rank). We obtain them for infinite fields like the complex numbers.

We prove our result more generally for a large class of functions on tensors, and in particular obtain similar properties for all functions in Strassen's asymptotic spectrum of tensors. We prove a variety of related structural results on the way. For instance, we prove that for any converging sequence of asymptotic ranks, the limit is also an asymptotic rank for some tensor. We leave open whether asymptotic rank is also discrete from below (which would be implied by Strassen's asymptotic rank conjecture).
\end{abstract}

\section{Introduction}
Asymptotic tensor rank is a fundamental parameter in algebraic complexity theory \cite{burgisser1997algebraic, blaser2013fast}. Originally rooted in the study of the matrix multiplication exponent, a burst of recent works has put this tensor parameter, and specifically Strassen's asymptotic rank conjecture, at the forefront of a wide range of computational complexity problems \cite{10.1145/3618260.3649656, 10.1145/3618260.3649620, björklund2024chromaticnumber19999ntime}.
These kinds of asymptotic tensor parameters more broadly play an important role in various fields, like additive combinatorics and quantum information theory (asymptotic slice rank, asymptotic subrank) \cite{tao, naslund_sawin_2017, vrana2015transformation}. 

Despite tremendous effort (resulting in new matrix multiplication algorithms \cite{le2014powers, DBLP:conf/soda/AlmanW21, duan2023faster, williams2023new}, barriers \cite{alman2018limits,v017a002,DBLP:conf/mfcs/BlaserL20}, new routes \cite{cohn2003group, cohn2005group,cohn2013fast,MR3631613,blasiak2022matrix, blasiak2024finitematrixmultiplicationalgorithms}, and fundamental theory \cite{DBLP:conf/focs/Strassen86, DBLP:conf/issac/Strassen12,MR3729273, MR4495838, wigderson2022asymptotic, kaski2024universalsequencetensorsasymptotic}), we are still far away from determining the matrix multiplication exponent, computing asymptotic ranks in general, or resolving the asymptotic rank conjecture. All in all, determining asymptotic rank has turned out very challenging and much is still unknown about the properties of this parameter.

In this paper we prove a \emph{polynomial characterization} of asymptotic tensor rank: for any number $r$ there are finitely many polynomials on tensors whose vanishing determines if the asymptotic rank is at most $r$ (just like matrix rank is characterized by vanishing of determinants of submatrices). This characterization has many consequences regarding the computation and topological structure of this parameter. Indeed, for any $r$ it leads to an algorithm to determine if asymptotic rank is at most $r$. We obtain from it that asymptotic rank is semi-continuous (like matrix rank) and that its values are well-ordered, that is, discrete from above: any non-increasing sequence of asymptotic ranks stabilizes. For the matrix multiplication exponent $\omega$, this implies in particular (as we will explain more) that there is a constant $\eps > 0$ such that no tensor has exponent between $\omega$ and $\omega + \eps$.

We will discuss these results and their meaning in more detail. First we discuss the context of this work in complexity theory and mathematics.

\subsubsection*{Matrix multiplication exponent and asymptotic rank conjecture}

It is a fundamental open problem to determine the matrix multiplication exponent $\omega$, which is defined as the infimum over all real numbers $\tau$ such that $n\times n$ matrices can be multiplied using $\mathcal{O}(n^\tau)$ arithmetic operations, with current state of the art $2 \leq \omega < 2.371339$ \cite{le2014powers, DBLP:conf/soda/AlmanW21, duan2023faster, williams2023new,almanMoreAsymmetryYields2025}. It is very well possible that $\omega = 2$. 

The matrix multiplication problem can naturally be phrased in terms of tensors and asymptotic rank. Namely $2^\omega = \asymprank(\langle 2,2,2\rangle)$, where the asymptotic rank $\asymprank(T) = \lim_{n\to\infty}\tensorrank(T^{\boxtimes n})^{1/n}$ of a tensor~$T$ measures the rate of growth of tensor rank on Kronecker powers of $T$, and where $\langle 2,2,2\rangle \in \FF^4 \otimes \FF^4 \otimes \FF^4$ is the so-called $2\times 2$ matrix multiplication tensor. Whether $\omega=2$ is thus tightly linked to the question whether asymptotic rank can take non-integer values or not. 

Strikingly, it is possible that for \emph{every} tensor $T \in \FF^n \otimes \FF^n \otimes \FF^n$ the asymptotic rank is at most~$n$, as opposed to the tensor rank, which can be $\Omega(n^2)$ (but we do not know such tensors explicitly \cite{DBLP:journals/jacm/Raz13, Blaser2014}).
Strassen's asymptotic rank conjecture indeed states that this is true, and more precisely states that asymptotic rank equals the largest flattening rank of the tensor (matrix rank after grouping two legs of the tensor), and would thus imply not only that asymptotic rank is always an integer, but also that it is easy to compute and that $\omega=2$. 

Intriguingly, there is a partial converse to the above connection between the asymptotic rank conjecture and the matrix multiplication exponent, namely \cite{strassen1988asymptotic}, for any tensor in $\FF^n \otimes \FF^n \otimes \FF^n$ the asymptotic rank is at most $n^{2\omega/3}$. In particular, if $\omega =2$, then every asymptotic rank is at most $n^{4/3}$, which is ``not far'' from the claim of the asymptotic rank conjecture. In this sense, matrix multiplication almost acts as a ``complete'' instance for the asymptotic rank conjecture. In the same spirit, Kaski and Michałek proved a completeness result for the asymptotic rank conjecture for explicit sequences of tensors \cite{kaski2024universalsequencetensorsasymptotic}.

Recent works have found new strong connections between the asymptotic rank conjecture and a range of computational complexity problems related to the chromatic number, set partitioning and the set cover conjecture \cite{10.1145/3618260.3649656, 10.1145/3618260.3649620, björklund2024chromaticnumber19999ntime}.

\subsubsection*{Structure of asymptotic ranks}

The asymptotic rank conjecture naturally leads to the question: What are the possible values that the asymptotic rank can take,
\[
\mathcal{R} = \{\asymprank(T) : T \in \FF^{d_1} \otimes \FF^{d_2} \otimes \FF^{d_3}, \, d_1, d_2, d_3 \in \ZZ_{\geq1}\}?
\]
And more generally: What is the structure (geometric, algebraic, topological, computational) of this set of values $\mathcal{R}$?
Is there anything we can say about $\mathcal{R}$ without resolving the asymptotic rank conjecture or determining $\omega$? Little is known. Clearly, such questions can be asked much more broadly, for higher-order tensors, but also for other asymptotic tensor parameters (asymptotic slice rank, subrank).

One known structural property is that $\mathcal{R}$ is closed under applying any univariate polynomial with nonnegative integer coefficients \cite[Theorem 4.8]{wigderson2022asymptotic}. Thus $\mathcal{R}$ has ``many'' elements. 
On the other hand, a sequence of recent works have revealed a strikingly ``discrete'' structure of sets like $\mathcal{R}$ (for various notions of asymptotic ranks), and thus that they have ``not too many'' elements. In particular, it can be seen (with a simple proof) that over any finite field (or finite set of coefficients in any field), $\mathcal{R}$ is a discrete set \cite{briet_et_al:LIPIcs.ITCS.2024.20}. Their techniques, however, do not apply to infinite fields like the complex numbers. This perhaps leaves the impression that discreteness may be a consequence of considering finitely many tensors in each format. 
(One of the main results of this paper, however, is that $\mathcal{R}$ is discrete from above over \emph{infinite} fields, like the complex numbers).

More broadly, notions of discreteness were proven for a range of asymptotic tensor parameters, $\asympfunc(T) = \lim_{n\to\infty} F(T^{\boxtimes n})^{1/n}$ for varying $F$.
It was shown in \cite{bdr} that for a class of functions over finite fields, which includes asymptotic (sub)rank and asymptotic slice rank, the set of values that they take is well-ordered (discrete from above). 
The work~\cite{MR4495838}, which resolved a problem on Strassen's asymptotic spectrum \cite{strassen1988asymptotic}, proved as a consequence that asymptotic slice rank over complex numbers takes only finitely many values per format (because of a deep result on the structure of moment polytopes in representation theory), and thus only countably many values in total.
The work \cite{https://doi.org/10.48550/arxiv.2212.12219} proved that a class of tensor parameters over complex numbers, which again includes asymptotic (sub)rank and asymptotic slice rank, take only countably many values.
\cite{briet_et_al:LIPIcs.ITCS.2024.20} proved that asymptotic (sub)rank and asymptotic slice rank over any finite set of coefficients take only a discrete set of values.
In \cite{cgz, gesmundo2023gap} explicit gaps were determined in the smallest values of these parameters.

Unlike for the asymptotic tensor rank (for which nothing is known), the computational complexity of tensor rank is well-understood. Namely, it is known that tensor rank is \textbf{NP}-hard over the rationals and \textbf{NP}-complete over any finite field~\cite{DBLP:journals/jal/Hastad90} and this was extended by~\cite{DBLP:journals/jacm/HillarL13}. Recently, \cite{DBLP:journals/mst/SchaeferS18, shitov2016hardtensorrank} proved the stronger statement that tensor rank over any field has the same complexity as deciding the existential theory of that field, and~\cite{DBLP:conf/approx/Swernofsky18, DBLP:conf/stoc/BlaserIJL18} proved that tensor rank is hard to approximate. These results, however, have no immediate implication for the computational complexity of asymptotic tensor rank (which, for all we know, may be in \textbf{P}).

\subsubsection*{New results}
In this paper:
\begin{itemize}    
    \item We prove that asymptotic tensor rank is computable from above: over ``computable fields'', for every upper bound $r$ there is an algorithm that, given any $d \times d \times d$ tensor~$T$, decides if its asymptotic tensor rank is at most $r$. %
    \item As the core ingredient, we prove, over any field $\FF$, that for a broad class of regularized tensor functionals, including all points in Strassen's asymptotic spectrum, the sublevel sets on every fixed tensor format are Zariski-closed. For tensors over subfields of~$\CC$, it follows that these functionals are lower-semicontinuous in the Euclidean topology. In particular, for asymptotic tensor rank, every sublevel set is determined by the vanishing of finitely many polynomials on tensors. 
    \item We prove for asymptotic tensor rank that the set of values it takes on all tensors is well-ordered (discrete from above), and that when~$\FF = \CC$ this set is closed in the Euclidean topology. In particular, in the context of the matrix multiplication exponent~$\omega$, we find that there is a constant $\eps > 0$ such that no tensor has asymptotic rank strictly between~$2^\omega$ and~$2^{\omega + \eps}$.
    \item For every point in Strassen's asymptotic spectrum, we prove that the set of values attained on tensors of arbitrary format is well-ordered. Via Strassen's duality, the Zariski-closedness of asymptotic spectrum sublevel sets further implies that for every tensor $T$, the set of tensors $S$ with $S^{\boxtimes n} \leq T^{\boxtimes (n + o(n))}$ (i.e., asymptotic restrictions of $T$) is Zariski-closed.
    \item As a technical ingredient for the general version of our result we develop new lower bounds on a type of tensor-to-matrix transformation (max-rank) that may be of independent interest.
\end{itemize}
\subsection{Polynomial characterization}

We prove that upper bounds on asymptotic rank are computable:
\begin{theorem}\label{th:intro-1}
  For any computable field $\FF$ and $k \geq3$, for any $r \in \RR$, there is an algorithm that, given $d\in \ZZ_{\geq 1}^k$ and any $k$-tensor $T \in \FF^{d_1} \otimes \dotsb \otimes \FF^{d_k}$, decides if $\asymprank(T) \leq r$.
\end{theorem}

Note that the algorithm in \cref{th:intro-1} is not uniform in~$r$, but is uniform in the dimensions~$d_i$ of the input tensor, for simple reasons that we will explain shortly.
Computability of the field here is meant in the sense of~\cite{rabinComputableAlgebraGeneral1960}.
The computability condition for the field is a natural one: we need to be able to write down (and compute with) the tensor and certain polynomials on the tensor space that the algorithm relies on. \cref{th:intro-1} almost directly follows from the following Zariski-closedness statement, which is the geometric core of the argument:

\begin{theorem}\label{th:intro-2}
For any field $\FF$,  $k\geq3$, $d \in \ZZ^k_{\geq1}$, and $r \in \RR$, the sublevel set 
\[
\{T \in \FF^{d_1} \otimes \cdots \otimes \FF^{d_k} : \asymprank(T) \leq r\}
\]
is Zariski-closed. 
\end{theorem}

In other words, for every $k, d, r$, there is a finite set of polynomials $p_1, \ldots, p_\ell$ on $V = \FF^{d_1} \otimes \cdots \otimes \FF^{d_k}$ such that for every $T \in V$ we have $\asymprank(T) \leq r$ if and only if all $p_j$ vanish on $T$. Over computable fields, like the algebraic closure of the rationals, these polynomials are also computable, and indeed lead to the algorithm of \cref{th:intro-1}:
\begin{proof}[Proof of~\cref{th:intro-1}]
  Let $r \in \RR$. Assume without loss of generality that~$r \geq 0$, otherwise~$\asymprank(T) > r$ always.
  By~\cref{th:intro-2}, there are finitely many polynomials $p_1, \dotsc, p_\ell$ on~$\FF^{\lfloor r \rfloor} \otimes \dotsb \otimes \FF^{\lfloor r \rfloor}$ whose vanishing set consists of the tensors of asymptotic rank at most~$r$.
  Then the algorithm is as follows. On input~$T \in \FF^{d_1} \otimes \dotsb \otimes \FF^{d_k}$, determine its flattening ranks~$\tensorrank_i(T)$. As these are lower bounds on asymptotic rank, if $\tensorrank_i(T) > r$ for any~$i$, then~$\asymprank(T) > r$.
  Otherwise, we can compute an embedding~$S$ of~$T$ into~$\FF^{\lfloor r \rfloor} \otimes \dotsb \otimes \FF^{\lfloor r \rfloor}$. Evaluate the polynomials~$p_1, \dotsc, p_\ell$ on~$S$. If all of them vanish, then~$\asymprank(T) \leq r$; otherwise,~$\asymprank(T) > r$.
\end{proof}

We note that \cref{th:intro-2} has ``practical'' implications: if for some family of tensors one can prove an upper bound on the asymptotic tensor rank, then by \cref{th:intro-2} this upper bound directly extends to the Zariski-closure. Before our result this was only known on finite unions of GL-orbits.\footnote{Indeed, it is known that for any $S$ in the closure of $(\GL_{d_1} \times \cdots \times \GL_{d_k}) T$ (i.e., $S$ is a degeneration of $T$) the asymptotic rank $\asymprank(S)$ is at most $\asymprank(T)$; and the closure of a finite union of orbits equals the union of the closures.} 
In particular, it also follows from \cref{th:intro-2} that, over the complex numbers, for any sequence of tensors $T_1, T_2, \ldots$ converging to a tensor $T$ in the Euclidean norm, if the asymptotic rank of all $T_i$ is at most $r$, then the asymptotic rank of $T$ is at most~$r$ (Euclidean lower-semicontinuity). For instance, we may apply this idea to tensors converging to a matrix multiplication tensor $\langle n,n,n\rangle$ in order to get upper bounds on $\omega$.

As an ingredient in the proof of \cref{th:intro-2} we prove that for any subset $A \subseteq \FF^{d_1} \otimes \cdots \otimes \FF^{d_k}$ the supremum of $\asymprank$ over the Zariski closure of $A$ equals the supremum over $A$ itself. This result we prove via a decomposition of powers of elements in the closure of $A$ in terms of powers of elements in $A$  %
combined with an asymptotic double-blocking analysis of asymptotic rank.

In fact, we prove \cref{th:intro-1} and \cref{th:intro-2} for a more general class of functions that includes all functions in Strassen's asymptotic spectrum of tensors. Via Strassen's duality~\cite{strassen1988asymptotic}, this implies that the set of all tensors that are an asymptotic restriction of a given tensor, is Zariski-closed.\footnote{For two tensors $S,T$ there is an asymptotic restriction $S \asympleq T$ if $S^{\boxtimes n} \leq T^{\boxtimes (n + o(n))}$. Strassen's duality says that $S \asympleq T$ if and only if for every $F$ in the asymptotic spectrum $\Delta(\FF,k)$ we have $F(S) \leq F(T)$. Then $\{S \in V : S \asympleq T\} = \cap_{F \in \Delta(\FF,k)} \{S : F(S) \leq F(T)\}$. Every sublevel set $\{S : F(S) \leq F(T)\}$ is closed, so their intersection is closed.} 

\subsection{Discreteness from above and closedness}

As a consequence of \cref{th:intro-2} we prove that the set of values that asymptotic rank takes, on all tensors (of fixed order, but arbitrary dimension), is well-ordered (discrete from above, every non-increasing sequence stabilizes):

\begin{theorem}\label{th:intro-3}
$\mathcal{R} = \{\asymprank(T) : T \in \FF^{d_1} \otimes \cdots \otimes \FF^{d_k}, d \in \ZZ_{\geq1}^k\}$ %
is well-ordered.
\end{theorem}

In particular, \cref{th:intro-3} says that there cannot be a sequence of tensors with asymptotic rank strictly larger than $2^\omega$ that gets arbitrarily close to it. Indeed, if it comes arbitrarily close it must eventually ``snap'' to $2^\omega$. Previously this was only known over finite fields (where the proof is simple, see \cite[Theorem 4.4]{briet_et_al:LIPIcs.ITCS.2024.20}; however that proof strategy does not work over infinite fields).  %

We leave as an open problem to prove discreteness from below for asymptotic rank. That is, we do not know if there can be (non-constant) increasing converging sequences of asymptotic ranks of tensors. 

Towards proving discreteness, as an intermediate result, we prove closedness of the set of values of asymptotic rank, when the base field is the complex numbers:

\begin{theorem}
Let~$\field = \CC$.
For any sequence in $\mathcal{R}$ that converges, the limit is in $\mathcal{R}$.
\end{theorem}

We note again that all of the above results are consistent with (and may be thought of as evidence towards) the asymptotic rank conjecture, since that would imply that asymptotic rank is computable as the matrix rank of a flattening of the tensor, and that the set of values that asymptotic rank takes is simply the natural numbers $\NN$.

\subsection{Asymptotic spectrum and tensor-to-matrix restrictions}

Finally, we extend \cref{th:intro-3} to the tensor parameters in Strassen's asymptotic spectrum of tensors \cite{strassen1988asymptotic, wigderson2022asymptotic}, a collection of tensor parameters with special properties (see \cref{sec:asympspec}). We denote by $\Delta(\FF,k)$ the asymptotic spectrum of $k$-tensors over $\FF$.
\begin{theorem}\label{th:intro:asympspec}
For every $F \in \Delta(\FF,k)$, $\{F(T) : T \in \FF^{d_1} \otimes \cdots \otimes \FF^{d_k}, d \in \ZZ_{\geq1}^k\}$ is well-ordered. %
When~$\FF = \CC$, this set is moreover Euclidean-closed in~$\RR$.
\end{theorem}

In fact, by Strassen's duality \cite{strassen1988asymptotic}, \cref{th:intro:asympspec} implies \cref{th:intro-3}, and moreover has discreteness implications for asymptotic transformation rates between tensors.
The proof of \cref{th:intro:asympspec} is more involved than \cref{th:intro-3}. It relies not only on the Zariski-closedness of sublevel sets, but also on a more technical ``growth'' argument (which may be of independent interest) in order to obtain well-orderedness across all tensor formats; the same growth argument also yields the Euclidean closedness over~$\CC$. %

This growth argument involves new lower bounds on a type of tensor-to-matrix restrictions.
In quantum information these correspond to $k$-partite to bipartite entanglement transformations (under stochastic local operations and classical communication, SLOCC).

For any $k$-tensor $T \in V_1 \otimes V_2 \otimes \cdots \otimes V_k$, let $\subrank_{1,2}(T)$ be the largest number $r$, such that there are linear $A_i\colon V_i \to V_i$ 
such that 
$
(A_1 \otimes \cdots \otimes A_k) T = \sum_{i=1}^r e_i \otimes e_i \otimes e_1 \otimes \cdots \otimes e_1.
$
The right hand side is essentially a rank $r$ identity matrix on tensor legs 1 and~2. This definition extends to $\subrank_{i,j}$ for $i\neq j \in [k]$ by placing this ``identity matrix'' at legs $i$ and $j$. %
For any proper subset $I \subseteq [k]$, let $\tensorrank_I(T)$ be the matrix rank of the matrix obtained by grouping the legs in $I$ together and grouping the remaining legs together. If $i \in I$ and $j \not\in I$, then $\tensorrank_I(T) \geq \subrank_{i,j}(T)$. We prove the following inequality in the opposite direction.

\begin{theorem}\label{th:intro-Qij}
Let $I \subseteq [k]$  with $1 \leq |I|\leq k-1$, and let $T$ be a $k$-tensor. If $|\FF| > \tensorrank_I(T)$, then 
$
\prod_{\substack{i \in I, j \in [k]\setminus I}} \subrank_{i,j}(T) \geq \tensorrank_I(T).
$
\end{theorem}
\cref{th:intro-Qij} generalizes \cite[Theorem~1.13]{briet_et_al:LIPIcs.ITCS.2024.20}, which covered $k=3$. Using \cref{th:intro-Qij} we prove that  any element $F$ in the asymptotic spectrum is either essentially an element in the asymptotic spectrum for lower order tensors, or grows with the size of the tensor (which in turn is the right ingredient for proving \cref{th:intro:asympspec}).

\subsection{Tensor preliminaries}\label{subsec:prelim}

\paragraph{Restriction, equivalence, conciseness, Kronecker product.}
All vector spaces considered in this paper are assumed to be finite-dimensional.
We denote by $V_1 \otimes \cdots \otimes V_k$ or $\FF^{d_1} \otimes \cdots \otimes \FF^{d_k}$ spaces of $k$-tensors. For two $k$-tensors $S$ and $T$ we say $S$ is a restriction of $T$ and write $S \leq T$ if there are linear maps $A_i$ such that $S = (A_1\otimes \cdots \otimes A_k)T$. Two tensors $S$ and $T$ are equivalent if $S \leq T$ and $T \leq S$. A tensor $T \in \FF^{d_1} \otimes \cdots \otimes \FF^{d_k}$ is called concise if for every $i \in [k]$, the flattening rank $\tensorrank_i(T) = \tensorrank_{\{i\}}(T)$ equals $d_i$. Any tensor is equivalent to a concise tensor. For two $k$-tensors $S$ and $T$ their Kronecker product $S \boxtimes T$ is the $k$-tensor obtained by taking the tensor product and grouping corresponding legs. 

\paragraph{Matrix multiplication, unit tensors and rank.}
For $a,b,c \in \ZZ_{\geq1}$, let $\langle a,b,c\rangle \in \CC^{ab} \otimes \CC^{bc} \otimes \CC^{ba}$ be the tensor $\sum_{i,j,k} e_{i,j} \otimes e_{j,k} \otimes e_{k,i}$ where the sum goes over $i \in [a], j \in [b], k \in [c]$. These are the matrix multiplication tensors and $\asymprank(\langle 2,2,2\rangle) = 2^\omega$. For $r \in \NN$, let $\langle r\rangle = \sum_{i=1}^r e_i \otimes \cdots \otimes e_i$. This is called the rank $r$ unit tensor (of order $k$). Tensor rank $\tensorrank(T)$ is defined as the smallest number $r$ such that $T$ can be written as a sum of $r$ simple tensors $u_1\otimes \cdots \otimes u_k$. We have $\tensorrank(T) = \min \{r \in \NN : T \leq \langle r\rangle\}$.

\section{Zariski-closed sublevel sets for asymptotic rank}
\label{sec:regularized-functionals-zariski-semicontinuous}

In this section we prove that the sublevel sets of asymptotic tensor rank are Zariski-closed, and hence characterized by the vanishing of finitely many polynomials. We then use this to prove that asymptotic rank is discrete from above. We work with a general class of regularizations of ``admissible functionals'' which includes asymptotic rank. 
Let~$\field$ be an arbitrary field.

\subsection{Admissible functionals}
To understand asymptotic rank we need to understand the behavior of tensor rank on powers of tensors. It turns out that our analysis can be done on a slightly more abstract level, where only some of the properties of the situation matter. In particular the space of tensors $\FF^{d_1} \otimes \cdots \otimes \FF^{d_k}$ we replace by an arbitrary vector space $V$ (recall that we assume vector spaces to be finite-dimensional), and we focus on families of functions defined on tensor powers of $V$, with properties that are familiar from tensor rank:
\begin{definition}[Admissible functional]
    \label{def:admissible functional}
    Let $V$ be a vector space over $\field$.
    An \emph{admissible functional}~$\func$ is a family of functions $\{\func_n \colon V^{\otimes n} \to \RR_{\geq 0}\}_{n \in \ZZ_{\geq1}}$ with the following properties:
    \begin{enumerate}[(i)]
        \item $\func_n(T + S) \leq \func_n(T) + \func_n(S)$ for $T, S \in V^{\otimes n}$,
        \item $\func_{n+n'}(T \otimes S) \leq \func_n(T) \func_{n'}(S)$ for $T \in V^{\otimes n}$ and $S \in V^{\otimes n'}$, 
        \item $\func_{n_1 + \dotsb + n_\ell}(T_1 \otimes \dotsb \otimes T_\ell) = \func_{n_1 +  \dotsb + n_\ell}(T_{\sigma(1)} \otimes \dotsb \otimes T_{\sigma(\ell)})$ for every permutation $\sigma : [\ell] \to [\ell]$ and tensors $T_j \in V^{\otimes n_j}$,
        \item $\func_n(\alpha T) = \func_n(T)$ for every non-zero $\alpha \in \field$ and every $T \in V^{\otimes n}$,
        \item %
        $\func_1$ is bounded: there exists $c$ such that for every $T\in V$, $F_1(T) \leq c$. 
    \end{enumerate}
    Note that \cref{def:admissible functional} implies that $F_n$ is bounded for every $n$.\footnote{In fact, $F_n(T)\le \sup_{S\in V} F_1(S)^n\tensorrank(T)\le(\sup_{S\in V} F_1(S)\dim V)^n$.}
    For any admissible functional~$\func$, and $T \in V$, define the \emph{regularization}
    \[
        \asympfunc(T) = \lim_{n \to \infty} \func_{n}(T^{\otimes n})^{1/n}.
    \]
    This is well-defined by Fekete's lemma, and the limit can be replaced by the infimum.
    For any subset $A \subseteq V$, we define
    \[
        \asympfunc[A] = \sup_{T \in A} \asympfunc(T).
    \]
\end{definition}

The main example of a regularized admissible functional is the asymptotic rank $\asymprank$, which we get by taking $V = \FF^{d_1} \otimes \cdots \otimes \FF^{d_k}$ to be a space of $k$-tensors and $\func_n$ to be the tensor rank on $V^{\otimes n}$ regrouped as $k$-tensors via $v_1 \otimes \cdots \otimes v_n \mapsto v_1 \boxtimes \cdots \boxtimes v_n$.\footnote{Here $\boxtimes$ is the ``grouped'' tensor product (Kronecker product), which takes two $k$-tensors and multiplies them to another $k$-tensor. However, one may also use the ``non-grouped'' tensor product instead to get an interesting admissible functional \cite{MR3760763, DBLP:journals/siaga/ChristandlGJ19}.}
More examples are given by the points in the asymptotic spectrum of tensors \cite{strassen1988asymptotic} (see \cref{sec:asympspec}), which includes the flattening ranks and the quantum functionals \cite{MR4495838}.

\subsection{Zariski-closed sublevel sets}

Let $V$ be a vector space over $\field$, and let $\FF[V]$ denote the polynomial functions on $V$.
Let $A \subseteq V$ be a subset. 
The Zariski-closure of $A$ is defined as
\begin{equation*}
    \overline{A} = \{ T \in V : \forall f \in \field[V], \, \left.f\right|_{A} \equiv 0 \Rightarrow f(T) = 0 \}.
\end{equation*}

\begin{theorem}
    \label{thm:regularized admissible functional semicontinuous}
    Let $\func$ be an admissible functional for $V$ and let $A \subseteq V$.
    Then 
    $
    \asympfunc[\overline{A}] = \asympfunc[A].
    $
\end{theorem}
Towards the proof of \cref{thm:regularized admissible functional semicontinuous} we use the following lemma.
We define, for any subset $A \subseteq V$ and $n \in \ZZ_{\geq1}$, the set $A^{(n)} = \{T^{\otimes n} : T \in A\}$. 

\begin{lemma}\label{lem:span arbitrary field}
    Let $A \subseteq V$. 
    Then $\overline{A}^{(n)} \subseteq \linspan A^{(n)}$.
\end{lemma}

\begin{proof}
The set $\linspan A^{(n)}$ is the intersection of all kernels $\ker \ell$ of linear forms $\ell$ on $\Sym^n V$ such that $\ell|_{A^{(n)}} \equiv 0$. For every such $\ell$, we will prove that $\overline{A}^{(n)} \subseteq \ker \ell$, which proves the claim. Define the function $f\colon V \to \FF$ by $f(T)= \ell(T^{\otimes n})$. Then $f$ is a polynomial function on $V$ and $f|_{A} \equiv 0$. By definition of the Zariski-closure also $f|_{\overline{A}} \equiv 0$. Then by definition of $f$ we have $\ell|_{\overline{A}^{(n)}} \equiv 0$.
\end{proof}

\begin{proof}[Proof of \cref{thm:regularized admissible functional semicontinuous}]
    From $A\subseteq \overline{A}$ it follows directly that $\asympfunc[A] \leq \asympfunc[\overline{A}]$. It remains to prove $\asympfunc[\overline{A}] \leq \asympfunc[A]$.
    Let $T \in \overline{A}$. By \cref{lem:span arbitrary field}, for every $n \in \ZZ_{\geq1}$, there are~$p(n) \in \ZZ_{\geq 1}$, $S_1, \ldots, S_{\lincombiparts(n)} \in A$ and $\alpha_1, \ldots, \alpha_{\lincombiparts(n)} \in \field$ such that $T^{\otimes n} = \sum_{i=1}^{\smash{\lincombiparts(n)}} \alpha_i S_i^{\otimes n}$, with the~$S_i^{\otimes n}$ linearly independent.
    Note that $\lincombiparts(n)$ grows at most polynomially in~$n$, since
    \[
        \lincombiparts(n) \leq \dim \linspan A^{(n)} \leq \dim \Sym^n(V) = \binom{\dim V + n - 1}{n}.
    \]
    Let $m \in \ZZ_{\geq1}$.
    Then 
    \[
    T^{\otimes nm} = \sum_{i_1, \ldots, i_m\in [\lincombiparts(n)]} \bigotimes_{j=1}^m \alpha_{i_j} S_{i_j}^{\otimes n}.
    \]
    By subadditivity and scalar invariance of $\func$,
    \[
        \func(T^{\otimes nm}) \leq \lincombiparts(n)^m \max_{i_1, \ldots, i_m \in [\lincombiparts(n)]} \func\Bigl(\bigotimes_{j=1}^m S_{i_j}^{\otimes n}\Bigr).
    \]
    Rearranging tensor factors, we have for some $m_i$ that sum to $m$, that
    \[
    \func\Bigl(\bigotimes_{j=1}^m S_{i_j}^{\otimes n}\Bigr) = \func\Bigl(\bigotimes_{i=1}^{\lincombiparts(n)} S_{i}^{\otimes m_i n}\Bigr).
    \]
    (Here and in the following, we implicitly omit factors for which $m_i=0$.)
    By submultiplicativity of $\func$,
    \[
    \func\Bigl(\bigotimes_{i=1}^{\lincombiparts(n)} S_{i}^{\otimes m_i n}\Bigr) \leq \prod_{i=1}^{\lincombiparts(n)} \func( S_i^{\otimes m_i n} ).
    \]
    For every $\eps > 0$, there is an $M(\eps,n) \in \ZZ_{\geq1}$ such that for every $\ell \geq M(\eps,n)$ and for every $i \in [\lincombiparts(n)]$, we have $\func(S_i^{\otimes \ell})^{1/\ell} \leq \asympfunc(S_i) + \eps$. By definition, $F_1$ is bounded, say $F_1 \leq B$ for some $B \in \ZZ_{\geq1}$. Then
    \[
    \prod_{i=1}^{\lincombiparts(n)} \func( S_i^{\otimes m_i n} ) \leq \prod_{\substack{i\in [\lincombiparts(n)]\\m_i n \geq M(\eps, n) }} \bigl( \asympfunc(S_i) + \eps\bigr)^{m_i n} \prod_{\substack{i\in [\lincombiparts(n)]\\ m_i n < M(\eps, n)}} B^{m_i n}
    \]
    Then
    \[
    \func(T^{\otimes nm})^{1/(nm)} \leq \lincombiparts(n)^{1/n} (\asympfunc[A] + \eps) B^{\lincombiparts(n) M(\eps, n) / (nm)}.
    \]
    We let $m$ go to infinity to get
    $\asympfunc(T) \leq \lincombiparts(n)^{1/n} (\asympfunc[A] + \eps)$.
    We let $\eps$ go to zero and $n$ go to infinity to get
    $\asympfunc(T) \leq \asympfunc[A]$.%
\end{proof}

\begin{corollary}\label{cor:general-lowerset-closed}
    Let $\func$ be an admissible functional for a vector space $V$.
    Let $r \in \RR$. 
    Then $\{ T \in V : \asympfunc(T) \leq r \}$ is Zariski-closed.
\end{corollary}
\begin{proof}
    Let $A = \{ T \in V : \asympfunc(T) \leq r \}$.
    Then $\asympfunc[A] \leq r$.
    By \cref{thm:regularized admissible functional semicontinuous}, $\asympfunc[\overline{A}] = \asympfunc[A]$, so that for all $T \in \overline{A}$, we have $\asympfunc(T) \leq r$, and hence $T \in A$.
    We conclude that $\overline{A} \subseteq A$.
\end{proof}

\subsection{Discreteness from above for asymptotic rank}

We call a subset $S \subseteq \RR$ well-ordered if every nonempty subset has a smallest element. Equivalently, any non-increasing sequence of elements in $S$ stabilizes.

\begin{corollary}\label{cor:asymprank-well-order}
    Let $\func$ be an admissible functional for a vector space $V$.
    Then  %
    $\{ \asympfunc(T) : T \in V \}$ is well-ordered. %
\end{corollary}
\begin{proof}
    Let $A_r = \{ T \in V : \asympfunc(T) \leq r \}$. Let $r_1 \geq r_2 \geq \cdots$ be a sequence of elements in $\{ \asympfunc(T) : T \in V \}$.
    Then $A_{r_1} \supseteq A_{r_2} \supseteq \cdots$ is a descending chain of Zariski-closed sets (\cref{cor:general-lowerset-closed}).
    Since the Zariski-topology on $V$ is Noetherian \cite[Section 1.4]{MR1322960}, the chain of~$A_{r_i}$ stabilizes, that is, there exists some $N \in \ZZ_{\geq 1}$ such that for all $n \geq N$, $A_{r_n} = A_{r_{n+1}}$.
    Let $T_N \in V$ such that $\asympfunc(T_N) = r_N$.
    Then for all $n \geq N$ we have $T_N \in A_{r_n}$.
    This implies that $r_n \geq r_N$, but we also have $r_n \leq r_N$, hence $r_n = r_N$ for all $n \geq N$.
\end{proof}

We can in particular apply \cref{cor:asymprank-well-order} to tensor rank to obtain well-orderedness of $\{\asymprank(T) : T \in V\}$ for any $V = \FF^{d_1} \otimes \cdots \otimes \FF^{d_k}$.
In fact, because asymptotic rank is lower bounded by the flattening ranks, the well-orderedness readily extends to all tensors (with fixed order, but arbitrary dimensions), as follows.

\begin{corollary}\label{cor:asymprank-wellord-alltensors}
The set $\{ \asymprank(T) : T \in \FF^{d_1} \otimes \cdots \otimes \FF^{d_k}, d \in \ZZ_{\geq1}^k \}$ is well-ordered.
\end{corollary}
\begin{proof}
Tensor rank is an admissible functional (\cref{def:admissible functional}), so 
by \cref{cor:asymprank-well-order}, the set $\{ \asymprank(T) : T \in \FF^{d_1} \otimes \cdots \otimes \FF^{d_k}\}$ is well-ordered for every $d \in \ZZ_{\geq1}^k$. Let $r_1 \geq r_2 \geq \cdots$ be a non-increasing %
sequence in $\{ \asymprank(T) : T \in \FF^{d_1} \otimes \cdots \otimes \FF^{d_k}, d \in \ZZ_{\geq1}^k \}$. This sequence is bounded from above (by $r_1$). For every $i$ there is a concise tensor $T_i$ (\cref{subsec:prelim}) such that $\asymprank(T_i) = r_i$, so 
that $T_i \in \FF^{d_1} \otimes \cdots \otimes \FF^{d_k}$ for some $d_1, \ldots, d_k \leq r_i$.
By re-embedding the~$T_i$, we see that the $r_i$ are all contained in $\{ \asymprank(T) : T \in \FF^{d} \otimes \cdots \otimes \FF^{d}\}$ where~$d = \lfloor r_1 \rfloor$, which is well-ordered.
So $r_1, r_2,\ldots$ stabilizes.
\end{proof}

\section{Discreteness from above for the asymptotic spectrum of tensors}\label{sec:asympspec}

In \cref{sec:regularized-functionals-zariski-semicontinuous} we proved discreteness from above for any admissible functional on any fixed vector space $V$, and for the asymptotic tensor rank on all tensors. In this section we will extend the latter result to a large class of admissible functionals, called asymptotic spectrum.

Let $\Delta(\FF, k)$ be the asymptotic spectrum of order-$k$ tensors over $\FF$ \cite{strassen1988asymptotic}. This is the set of all functions from $k$-tensors over $\FF$ to reals that are $\boxtimes$-multiplicative, $\oplus$-additive, normalized to~$r$ on $\langle r\rangle$ and monotone under restriction. We note that asymptotic tensor rank is not an element of $\Delta(\FF,k)$. However, by Strassen's duality theorem \cite{strassen1988asymptotic}, it is the pointwise maximum over it, and in fact the asymptotic spectrum characterizes many asymptotic properties of tensors (asymptotic restriction, asymptotic subrank, asymptotic slice rank) \cite{MR4495838}.

Any element of $\Delta(\FF, k)$ gives an admissible functional on $V = \FF^{d_1} \otimes \cdots \otimes \FF^{d_k}$ (\cref{def:admissible functional}). Moreover, for every $F \in \Delta(\FF, k)$, it follows from multiplicativity that $F$ coincides with its regularization~$\asympfunc$. 
Thus, for every $F$ in the asymptotic spectrum of tensors, the following is true: 
\begin{enumerate}[(1)]
\item For every $A \subseteq V$ we have $F(\overline{A}) = F(A)$ (\cref{thm:regularized admissible functional semicontinuous}). 
\item For every $r \in \RR$, the set $\{T \in V : F(T) \leq r\}$ is Zariski-closed (\cref{cor:general-lowerset-closed}). \item $\{F(T) : T \in V\}$ is well-ordered (\cref{cor:asymprank-well-order}).
\end{enumerate}
In this section we will extend the well-orderedness property (3) from all finite-dimensional spaces $V = \FF^{d_1} \otimes \cdots \otimes \FF^{d_k}$ to the space of all $k$-tensors (analogously to \cref{cor:asymprank-wellord-alltensors} for asymptotic rank):

\begin{theorem}\label{th:spec-well-ord}\label{thm:spectral points well ordered}
For any field $\FF$, any order $k \in \ZZ_{\geq2}$ and any function $F \in \Delta(\FF, k)$, %
the set of values $\{F(T) : T \in \FF^{d_1} \otimes \cdots \otimes \FF^{d_k}, d \in \ZZ_{\geq1}^k\}$ is well-ordered.
\end{theorem}

\subsection{Tensor-to-matrix transformations}
Before proving \cref{th:spec-well-ord} we prove an important ingredient for it.
For $r \in \NN$, let $\langle r\rangle_{1,2}$ denote the $k$-tensor $\sum_{\ell=1}^r e_\ell \otimes e_\ell \otimes e_1 \otimes \cdots \otimes e_1$, and similarly define $\langle r\rangle_{i,j}$ with $i\neq j\in [k]$ indicating the location of the factors $e_\ell$ in the tensor product. Let $\subrank_{i,j}(T)$ be the largest number $r$ such that $T \geq \langle r\rangle_{i,j}$. Thus $\subrank_{i,j}(T)$ essentially measures the largest matrix rank of a matrix (2-tensor) that can be obtained on legs $i, j$ via restriction.\footnote{Note that for $k=3$, $\langle r\rangle_{1,2}$ is the matrix multiplication tensor $\langle 1,r,1\rangle$. In quantum information theory, the unit tensor $\langle r\rangle$ is known as a rank-$r$ Greenberger--Horne--Zeilinger (GHZ) state, and the tensor $\langle r\rangle_{1,2}$ corresponds to a rank-$r$ Einstein--Podolski--Rosen (EPR) pair between parties 1 and 2.}

We prove the following lower bound (extending the $k=3$ version from \cite{briet_et_al:LIPIcs.ITCS.2024.20}) that we will use for our discreteness result \cref{th:spec-well-ord}. Recall that $\tensorrank_I(T)$ denotes the flattening rank along~$I\subseteq [k]$: the matrix rank of the matrix obtained from $T$ by grouping legs in $I$ together and grouping legs in $[k] \setminus I$ together.

\begin{theorem}\label{th:Qij}
Let $I \subseteq [k]$  with $1\leq |I|\leq k-1$, and let $T$ be a $k$-tensor. If $|\FF| > \tensorrank_I(T)$, then 
\[
\prod_{\substack{i \in I\\ j \in [k]\setminus I}} \subrank_{i,j}(T) \geq \tensorrank_I(T).
\]
\end{theorem}

As a direct consequence of \cref{th:Qij} we see that at least one of the $\subrank_{i,j}$ is ``large'':

\begin{corollary}\label{cor:Qij}
For any $k$-tensor $T$, for any $i \in [k]$, assuming $|\FF| > \tensorrank_i(T)$, there is a $j \in [k] \setminus \{i\}$ such that $\subrank_{i,j}(T) \geq \tensorrank_i(T)^{1/(k-1)}$.
\end{corollary}
\begin{proof}
Apply \cref{th:Qij} with $I = \{i\}$, and note that one of the $k-1$ terms in the product must satisfy the claim.
\end{proof}

For the proof of \cref{th:Qij} we use the following lemma, which says that if we have a collection of matrices, then after applying a simultaneous linear operation on the columns, the linearly independent columns are ``flushed'' to the left in the matrices. This is proven using a Schwartz--Zippel type argument. We denote, for matrices $M_1, \ldots, M_t \in \FF^{a \times b}$, by $[M_1; \cdots ; M_t]$ the matrix in $\FF^{a \times (t b)}$ obtained by concatenating the matrices $M_i$ horizontally. For a matrix $M$ we denote by $M|_{S}$ the submatrix of $M$ with columns indexed by $S$.

\begin{lemma}[\cite{briet_et_al:LIPIcs.ITCS.2024.20}]\label{lem:flush}
Let $M_1, \ldots, M_c \in \FF^{a \times b}$ be matrices, and for any $i \in [c]$ let \[
r_i = \tensorrank([M_1; \cdots; M_i]) - \tensorrank([M_1; \cdots ; M_{i-1}]).
\]
If $|\FF|>\tensorrank([M_1; \cdots; M_c])$, then there exists an invertible matrix $U \in \FF^{b \times b}$ such that
\[
\textnormal{colspan}([(M_1U)|_{[r_1]}; \cdots; (M_c U)|_{[r_c]}]) = \textnormal{colspan}([M_1;\cdots;M_c]).
\]
\end{lemma}
In the proof we need the following. 
For any $\ell \in [k]$ and $i \in [d_\ell]$ we let 
\[
P^{(\ell)}_i : \FF^{d_1} \otimes \cdots \otimes \FF^{d_k} \to \FF^{d_1} \otimes \cdots \otimes \widehat{\FF^{d_\ell}} \otimes \cdots \otimes \FF^{d_k}
\]
be the projector from $k$-tensors to $(k-1)$-tensors that applies $e^*_i$ to the $\ell$th tensor leg, where~$\widehat{\FF^{\smash{d_\ell}}}$ indicates that we leave out the factor $\FF^{d_\ell}$.
For any subset $L \subseteq [k]$ and $i \in \prod_{\ell \in L} [d_\ell]$ (Cartesian product) we let $P^{\smash{(L)}}_i$ be the projector from $k$-tensors to $(k-|L|)$-tensors that applies $e_{i_1}^*\otimes\dots\otimes e_{i_L}^*$ to the tensor legs indexed by $L$.
We call $P^{\smash{(L)}}_i T$ the $i$th $L$-slice of $T$. If we think of $T$ as a $k$-dimensional array of elements of $\FF$ indexed by $[d_1] \times \cdots \times [d_k]$, then $P^{\smash{(L)}}_i T$ is the array obtained from $T$ by fixing the indices labeled by $L$ to the index vector $i$.

\begin{proof}[Proof of \cref{th:Qij}]
The proof is by induction on the order $k$. In the base case $k = 2$, the functions $\subrank_{1,2}$ and $\tensorrank_{\{1\}}$ coincide (they are both the matrix rank), so the claim is true. We now assume that the claim holds for the cases $2, \ldots, k-1$ and prove it for $k$. 
Let $I \subseteq [k]$ with $1 \leq |I| \leq k-1$. Since $\tensorrank_I = \tensorrank_{[k]\setminus I}$ we may assume $|I|\leq \lfloor k/2\rfloor$. After a permutation of the tensor legs, we may then without loss of generality take $I = \{1,\ldots, m\}$ for some $1 \leq m \leq \lfloor k/2\rfloor$. 
Let $T \in \FF^{n_1} \otimes \cdots \otimes \FF^{n_k}$. 
The $I$-flattening $T_I$ is an $(n_1\cdots n_m) \times (n_{m+1} \cdots n_k)$ matrix. We write~$T_I$ as a block matrix
\[
T_I = [M_{1, \ldots, 1}; \cdots ; M_{i_{m+1}, \ldots, i_{k-1}}; \cdots ; M_{n_{m+1}, \ldots, n_{k-1}}],
\]
where $i = (i_{m+1}, \ldots, i_{k-1})$ goes over $[n_{m+1}] \times \cdots \times [n_{k-1}]$ (say in lexicographic order) and 
where each block $M_{i_{m+1}, \ldots, i_{k-1}}$ is an $(n_1 \cdots n_m) \times n_k$ matrix. For each $i \in [n_{m+1}] \times \cdots \times [n_{k-1}]$, let $r_i$ be the number of linearly independent columns of $M_i$ that are also linearly independent of all the previous blocks $M_{i'}$ in~$T_I$. Then $\tensorrank_I(T)$ equals the sum of $r_i$ over all $i \in [n_{m+1}] \times \cdots \times [n_{k-1}]$.
Let $L = [k-1]\setminus I = \{m+1, \ldots, k-1\}$.
We will prove two inequalities:
\begin{equation}\label{eq:1}
\prod_{j \in I} \subrank_{j,k}(T) \geq \max \{ r_i : i \in [n_{m+1}] \times \cdots \times [n_{k-1}] \}
\end{equation}
and
\begin{equation}\label{eq:2}
\prod_{j\in I, \ell\in L} \subrank_{j,\ell}(T) \geq |\{ i \in [n_{m+1}] \times \cdots \times [n_{k-1}] : r_i \neq 0\}|.
\end{equation}
From \eqref{eq:1} and \eqref{eq:2} the claim follows. Indeed, the product in the left-hand side of the claim splits into two products to which we apply the inequalities \eqref{eq:1} and \eqref{eq:2}:
\begin{align*}
\prod_{\substack{j \in I\\ \ell \in [k]\setminus I}} \subrank_{j,\ell}(T) 
&=  \Bigl( \prod_{j \in I} \subrank_{j,k}(T) \Bigr) \cdot \Bigl( \prod_{j \in I, \ell \in L} \subrank_{j,\ell}(T) \Bigr) \\[-0.4em]
&\geq  \max \{ r_i : i \in  [n_{m+1}] \times \cdots \times [n_{k-1}] \}  \cdot | \{ i \in [n_{m+1}] \times \cdots \times [n_{k-1}] :r_i\neq 0\}|\\
&\geq \textstyle\sum_{i \in [n_{m+1}] \times \cdots \times [n_{k-1}]} r_i \\%
&=\tensorrank_I(T).
\end{align*}

It remains to prove \eqref{eq:1} and \eqref{eq:2}. We first prove \eqref{eq:1}.
For any $i \in [n_{m+1}] \times \cdots \times [n_{k-1}]$, recall that we denote by $P^{\smash{(L)}}_i T$ the $i$th $L$-slice of $T$, which is an $(m+1)$-tensor of shape $n_1 \times \cdots \times n_m \times n_k$. The $k$-flattening of $P^{\smash{(L)}}_i T$ is the $(n_1 \cdots n_m) \times n_k$ matrix $M_i$. %
Since $M_i$ is a submatrix of $T_I$, we have $\tensorrank_I(T) \geq \tensorrank_k(P^{\smash{(L)}}_i T)$ and thus $|\FF| > \tensorrank_k(P^{\smash{(L)}}_i T)$. Thus by the induction hypothesis applied to the $(m+1)$-tensor $P^{\smash{(L)}}_i T$ we have (noting that $m+1 \leq k-1$ and that $[k]\setminus L = \{k\}$)
\begin{equation}\label{eq:11a}
\prod_{j \in I} \subrank_{j,k}(P^{(L)}_i T) \geq \tensorrank_k(P^{(L)}_i T) = \tensorrank(M_i).
\end{equation}
Let $i'$ such that $r_{i'} = \max \{ r_i : i \in [n_{m+1}] \times \cdots \times [n_{k-1}]\}$.
Then 
\begin{equation}\label{eq:11b}
\tensorrank(M_{i'}) \geq r_{i'}= \max \{ r_i : i \in [n_{m+1}] \times \cdots \times [n_{k-1}]\}.
\end{equation}
From \eqref{eq:11a} and \eqref{eq:11b} we get
\begin{equation}%
\prod_{j \in I} \subrank_{j,k}(T) \geq \prod_{j \in I} \subrank_{j,k}(P^{(L)}_{i'} T) \geq \max \{ r_i : i \in [n_{m+1}] \times \cdots \times [n_{k-1}]\}.
\end{equation}
This proves \eqref{eq:1}.

Now we will prove \eqref{eq:2}.
Since $|\FF| > \tensorrank_I(T)$, we may apply \cref{lem:flush} to the blocks $M_i$ in $T_I$ (with $a = n_1\cdots n_m$, $b = n_k$ and $c = n_{m+1}\cdots n_{k-1}$) to ``flush'' all linearly independent columns in $T_I$ to the left of each block $M_i$ by applying an invertible matrix $U$. We may assume that $T$ already has this \emph{column flushed} property (since this application of $U$ corresponds to applying $U$ to the $k$th tensor leg of $T$).

Let $P^{\smash{(k)}}_1 T$ be the first $k$-slice of $T$, which is a $(k-1)$-tensor of shape $n_1 \times \cdots \times n_{k-1}$.
The $L$-flattening of $P^{\smash{(k)}}_1 T$ is a matrix $M'$ for which the columns are precisely each first column of the matrices $M_i$ over all $i \in [n_{m+1}] \times \cdots \times [n_{k-1}]$. In particular, we have $\tensorrank_I(T) \geq \tensorrank_L(P^{\smash{(k)}}_1 T)$ and thus $|\FF| > \tensorrank_L(P^{\smash{(k)}}_1 T)$. %
We then have, by the induction hypothesis applied to the $(k-1)$-tensor~$P^{\smash{(k)}}_1 T$ (noting that $[k-1]\setminus I = L$), %
\begin{equation}\label{eq:12a}
\prod_{j \in I, \ell \in L} \subrank_{j,\ell}(P^{(k)}_1 T) \geq \tensorrank_L(P^{(k)}_1 T) = \tensorrank(M').
\end{equation}
Furthermore, by the column flushed property, we have
\begin{equation}\label{eq:12b}
\tensorrank(M') \geq |\{ i \in [n_{m+1}] \times \cdots \times [n_{k-1}] : r_i \neq 0\}|.
\end{equation}
From \eqref{eq:12a} and \eqref{eq:12b} we get
\begin{equation}%
\prod_{j\in I, \ell\in L} \subrank_{j,\ell}(T) \geq \prod_{j\in I, \ell\in L} \subrank_{j,\ell}(P^{(k)}_1 T) \geq |\{ i \in [n_{m+1}] \times \cdots \times [n_{k-1}] : r_i \neq 0\}|.
\end{equation}
This proves \eqref{eq:2}.
\end{proof}

As a consequence of \cref{th:Qij} we obtain a lower bound on the asymptotic subrank. (We will not use this lower bound here, but we expect it to be of independent interest.) The subrank $\subrank(T)$ is the largest number $r$ such that $T \geq \langle r\rangle$. The asymptotic subrank is defined as $\asympsubrank(T) = \lim_{n\to\infty} \subrank(T^{\boxtimes n})^{1/n}$. The limit exists and can be replaced by a supremum (Fekete's lemma).

\begin{corollary}
Let $T$ be a $k$-tensor over any field $\FF$. Then
\[
\asympsubrank(T) \geq \min_{\emptyset \neq I\subsetneq [k]} \tensorrank_I(T)^{2/(k(k-1))}.
\]
\end{corollary}
\begin{proof}
Let $q_{i,j} = \subrank_{i,j}(T)$ for every $i,j \in [k]$ with $i < j$. Then 
\[
T^{\boxtimes (k(k-1)/2)} \geq \bigkron_{i<j} \langle q_{i,j} \rangle_{i,j}.
\]
It was proven in \cite{MR3627407} that
\[
\asympsubrank\bigl(\bigkron_{i<j} \langle q_{i,j} \rangle_{i,j}\bigr) = \min_{\emptyset \neq I \subsetneq [k]} \prod_{\substack{i\in I\\ j \in [k]\setminus I}} q_{i,j}.
\]
Since the asymptotic subrank is invariant under field extension \cite[Theorem 3.10]{strassen1988asymptotic}, we may assume that $\FF$ is infinite. By \cref{th:Qij} we then have
\[
\asympsubrank(T^{\boxtimes (k(k-1)/2)}) \geq \min_{\emptyset \neq I \subsetneq [k]} \tensorrank_I(T).
\]
Taking the $k(k-1)/2$-th root gives the claim.
\end{proof}

\subsection{Discreteness from above for the asymptotic spectrum}

For tensors $S, T$ we write $S \sim T$, and say $S$ and $T$ are restriction-equivalent, if $S \leq T$ and $T \leq S$.

\begin{lemma}\label{lem:free-units}
Let $F \in \Delta(\FF, k)$ and $i\neq j \in [k]$. If $F(\langle 2\rangle_{i,j}) = 1$, then $F(\langle m\rangle_{i,j}) = 1$ for every $m \geq 1$.
\end{lemma}
\begin{proof}
For any $a \in \ZZ_{\geq 0}$, by multiplicativity, $F(\langle 2^a\rangle_{i,j}) = F(\langle 2\rangle_{i,j}^{\boxtimes a}) = F(\langle 2\rangle_{i,j})^a = 1$. For general $m \geq 1$, choose $a$ with $m \leq 2^a$. Then $\langle 1\rangle_{i,j} \leq \langle m\rangle_{i,j} \leq \langle 2^a\rangle_{i,j}$, so by monotonicity $1 = F(\langle 1\rangle_{i,j}) \leq F(\langle m\rangle_{i,j}) \leq F(\langle 2^a\rangle_{i,j}) = 1$.
\end{proof}

Fix $k \geq 3$. Let $\gamma$ be the linear map from $k$-tensors to $(k-1)$-tensors defined on simple tensors by
$\gamma(v_1 \otimes \cdots \otimes v_k) := v_1 \otimes \cdots \otimes v_{k-2} \otimes (v_{k-1} \otimes v_k)$. 
Let $\phi$ be the map from $(k-1)$-tensors to $k$-tensors defined by $\phi(U) := U \otimes e_1$.

\begin{lemma}\label{lem:free-splitting}
Let $F \in \Delta(\FF, k)$ and suppose $F(\langle 2\rangle_{k-1,k}) = 1$. Then for any two $k$-tensors $S, T$, if $\gamma(S) \sim \gamma(T)$, then $F(S) = F(T)$.
\end{lemma}
\begin{proof}
Since $\phi$ preserves restriction, $\gamma(S) \sim \gamma(T)$ implies $\phi(\gamma(S)) \sim \phi(\gamma(T))$, so by monotonicity of $F$, $F(\phi(\gamma(S))) = F(\phi(\gamma(T)))$. It thus suffices to show that $F(\phi(\gamma(T))) = F(T)$ for every $k$-tensor~$T$. For sufficiently large $m$, we have $\phi(\gamma(T)) \boxtimes \langle m\rangle_{k-1,k} \geq T$ and $\phi(\gamma(T)) \leq T \boxtimes \langle m\rangle_{k-1,k}$. Applying~$F$ and using multiplicativity, monotonicity, and $F(\langle m\rangle_{k-1,k}) = 1$ (\cref{lem:free-units}), we obtain $F(\phi(\gamma(T))) \geq F(T)$ and $F(\phi(\gamma(T))) \leq F(T)$, hence $F(\phi(\gamma(T))) = F(T)$.
\end{proof}

\begin{lemma}\label{lem:spec-descend}
Let $F \in \Delta(\FF, k)$ and suppose $F(\langle 2\rangle_{k-1,k}) = 1$. Then $F \circ \phi \in \Delta(\FF, k-1)$ and $F = F \circ \phi \circ \gamma$. In particular, $F$ and $F \circ \phi$ take the same set of values.
\end{lemma}
\begin{proof}
We first check that $F \circ \phi \in \Delta(\FF, k-1)$. Let $S, T$ be $(k-1)$-tensors. For multiplicativity: $\phi(S) \boxtimes \phi(T) \sim \phi(S \boxtimes T)$, so $F(\phi(S \boxtimes T)) = F(\phi(S) \boxtimes \phi(T)) = F(\phi(S))\,F(\phi(T))$. For additivity: $\gamma(\phi(S \oplus T)) \sim \gamma(\phi(S) \oplus \phi(T))$, so by \cref{lem:free-splitting}, $F(\phi(S \oplus T)) = F(\phi(S) \oplus \phi(T)) = F(\phi(S)) + F(\phi(T))$. For normalization: $\gamma(\phi(\langle r\rangle)) \sim \gamma(\langle r\rangle)$, where the first $\langle r\rangle$ is the rank-$r$ unit tensor of order $k-1$ and the second is the rank-$r$ unit tensor of order $k$, so by \cref{lem:free-splitting}, $F(\phi(\langle r\rangle)) = F(\langle r\rangle) = r$. For restriction-monotonicity: $\phi$ preserves restriction, so $S \leq T$ implies $F(\phi(S)) \leq F(\phi(T))$.

For the identity $F = F \circ \phi \circ \gamma$: for any $k$-tensor $T$, $\gamma(\phi(\gamma(T))) \sim \gamma(T)$, so by \cref{lem:free-splitting}, $F(\phi(\gamma(T))) = F(T)$. The set-of-values claim follows.
\end{proof}

\begin{proof}[Proof of \cref{th:spec-well-ord}]
Without loss of generality we may assume the field $\FF$ to be infinite, since for any field extension $\FF \subseteq \FF'$, any element in the asymptotic spectrum $\Delta(\FF,k)$ has a (unique) extension to $\Delta(\FF',k)$ \cite[Theorem 3.10]{strassen1988asymptotic}, so we may work there.  
(We will use this assumption to apply \cref{cor:Qij}.)

The proof is by induction on $k$. In the base case $k=2$, $\Delta(\FF,2)$ contains only matrix rank, for which the claim is true. Suppose $k > 2$. Let $F \in \Delta(\FF, k)$. 

Suppose $F(\langle 2\rangle_{i,j}) > 1$ for every $i\neq j \in [k]$. Then it follows from \cref{cor:Qij} and monotonicity of $F$ under restriction, that for every $k$-tensor $T$, $F(T)$ is at least some power of $\max_i \tensorrank_i(T)$. Thus if $r_1, r_2, \ldots$ is a non-increasing converging sequence in $\{F(T) : T \in \FF^{d_1} \otimes \cdots \otimes \FF^{d_k}, d \in \ZZ_{\geq1}^k\}$, then in fact this sequence is contained in $\{F(T) : T \in \FF^{d_1} \otimes \cdots \otimes \FF^{d_k}\}$ for some $d \in \ZZ_{\geq1}^k$ and thus stabilizes (\cref{cor:asymprank-well-order}).

Suppose that there are $i\neq j \in [k]$ such that $F(\langle 2\rangle_{i,j}) = 1$. By relabeling legs we may take $i = k-1$ and $j = k$. By \cref{lem:spec-descend}, there is a spectral point in $\Delta(\FF, k-1)$ that takes the same set of values as $F$. By the induction hypothesis, this set of values is well-ordered.
\end{proof}

\begin{remark}
There are interesting connections between the above results and recent work on topological Noetherianity of~$\GL_\infty$-varieties \cite{draismaTopologicalNoetherianityPolynomial2019, bdr}. Indeed, our~\cref{cor:general-lowerset-closed} combined with %
results in those works leads to an alternative proof of \cref{th:spec-well-ord}.
\end{remark}

\section{Towards discreteness from below} %

\cref{cor:asymprank-wellord-alltensors} leaves open whether asymptotic rank is discrete from below. In this section, we go into directions towards proving that. In the first part we prove closedness, in the Euclidean topology, of the set of values that asymptotic rank takes, and extend the same statement to every point in Strassen's asymptotic spectrum. In the second part we prove an equivalent geometric characterization of discreteness from below.

\subsection{Euclidean closedness over the complex numbers}
In this section our base field is the complex numbers.
In \cref{sec:regularized-functionals-zariski-semicontinuous} we proved that the set $\mathcal{R} = \{ \asymprank(T) : T \in \CC^{d_1} \otimes \cdots \otimes \CC^{d_k}, d \in \ZZ^k_{\geq1} \}$ is well-ordered (\cref{cor:asymprank-well-order}). In particular, for any decreasing, and thus converging, sequence of elements in $\mathcal{R}$, the limit is in $\mathcal{R}$. %
In this section, we will prove that for \emph{any} converging sequence in~$\mathcal{R}$ the limit is in $\mathcal{R}$, that is:

\begin{theorem}\label{cor:closed}
    $\{ \asymprank(T) : T \in \CC^{d_1} \otimes \cdots \otimes \CC^{d_k}, d \in \ZZ^k_{\geq1} \}$ is 
    Euclidean-closed in~$\RR$.
\end{theorem}

To prove \cref{cor:closed}, a key lemma is the following ``Baire property'' for affine varieties over~$\CC$. Let~$V$ be a complex vector space.

\begin{lemma}
    \label{lem:baire-property}
    Let $X \subseteq V$ be non-empty and Zariski-closed.
    The intersection of any countable collection of Zariski-open and Zariski-dense subsets of $X$ is Zariski-dense in $X$.
\end{lemma}
\begin{proof}
    For $i \in \NN$, let $U_i \subseteq X$ be Zariski-open and Zariski-dense.
    Then $U_i$ is Euclidean-dense in $X$~\cite[Prop.~5]{serreGeometrieAlgebriqueGeometrie1956}.
    Also the $U_i$ are Euclidean-open, because of the general fact that Zariski-open implies Euclidean-open.
    With the Euclidean topology, $X$ is a complete metric space and thus a Baire space (by the Baire category theorem~\cite[Thm.~48.1]{munkres2000topology}), which implies that the intersection $\cap_i U_i$ is Euclidean-dense in~$X$. Now we use the general fact that Euclidean-dense implies Zariski-dense.
\end{proof}

We briefly work in a more general setting.
Let~$V$ be a complex vector space. A function~$\func\colon V \to \RR$ is called Zariski-lower-semicontinuous if for every~$r \in \RR$, the set~$\{T \in V : \func(T) \leq r \}$ is Zariski-closed.

For $A \subseteq V$ we write $\func[A] = \sup_{T \in A} \func(T)$.
\begin{theorem}
    \label{thm:irreducible-variety-rank-maximizers-dense-general}
    Let~$X \subseteq V$ be non-empty, Zariski-closed and irreducible, and let~$F\colon V \to \RR$ be Zariski-lower-semicontinuous.
    Then the set $\{ T \in X : \func(T) = \func[X] \}$ is Zariski-dense in $X$. In particular, it is non-empty, and~$\func[X] < \infty$.
\end{theorem}
\begin{proof}
    Let $X_{\leq r} = \{ T \in X : \func(T) \leq r \}$ and $X_{> r} = \{ T \in X : \func(T) > r \}$.
    Then $X_{\leq r}$ is Zariski-closed in $X$ (by assumption of $\{ T \in V : \func(T) \leq r \}$ and $X$ being Zariski-closed), and $X_{>r}$ is Zariski-open.
    Since $X$ is irreducible, whenever $X_{\leq r}$ is a proper subset of $X$, its complement $X_{> r}$ is a non-empty Zariski-open subset and thus Zariski-dense (as every non-empty Zariski-open is Zariski-dense in an irreducible variety).
    Let $r_1 < r_2 < \cdots < F[X]$ be a sequence of real numbers such that~$\sup_{i \geq 1} r_i = F[X]$.
    Since $\{ T \in X : \func(T) = \func[X] \}$ equals the countable intersection $\bigcap_{i \geq 1} X_{>r_i}$, and since every $X_{>r_i}$ is nonempty (because $r_i < F[X]$), the claim follows by \cref{lem:baire-property}.
\end{proof}
\begin{corollary}\label{cor:exist-max-general}
    Let $X \subseteq V$ be non-empty and Zariski-closed, and let~$\func\colon V \to \RR$ be Zariski-lower-semicontinuous.
    Then there exists $T \in X$ such that $\func(T) = \func[X]$.
\end{corollary}
\begin{proof}
    Let $X = X_1 \cup \dotsb \cup X_\ell$ be the decomposition of $X$ into irreducible components.
    Then $\func[X] = \max_{j \in [\ell]} \func[X_j]$, and \cref{thm:irreducible-variety-rank-maximizers-dense-general} yields for every irreducible component a tensor $T_j \in X_j$ such that $\func(T_j) = \func[X_j]$.
    Hence $\func[X] = \func(T_j)$ for some $j \in [\ell]$.
\end{proof}
\begin{corollary}\label{cor:eucl-cl-general}
    Let~$\func\colon V \to \RR$ be Zariski-lower-semicontinuous.
    Then~$\{\func(T) : T \in V \}$ is closed (in the Euclidean topology on $\RR$).
\end{corollary}
\begin{proof}
    This set is well-ordered (as a consequence of the Zariski-closedness of sublevel sets of $\func$, combined with topological Noetherianity of the Zariski-topology, cf.~the proof of \cref{cor:asymprank-well-order}), so every decreasing converging sequence in $\mathcal{\func} = \{\func(T) : T \in V \}$ has a limit in $\mathcal{\func}$. It remains to prove the same for every increasing sequence.
    Let $r_1 < r_2 < \cdots$ be a sequence in $\mathcal{\func}$ converging to~$r \in \RR$.
    Let $V_{\leq r} = \{ T \in V : \func(T) \leq r \}$. Then $V_{\leq r}$ is Zariski-closed.
    We have $\func[V_{\leq r}] \leq r$ by definition, and $\func[V_{\leq r}] \geq r_i$ for every $i \geq 1$, so $\func[V_{\leq r}] = r$.
    Therefore there exists $T \in V_{\leq r}$ such that $\func(T) = \func[V_{\leq r}] = r$ (\cref{cor:exist-max-general}). We conclude that $r \in \mathcal{\func}$.
\end{proof}

We now return to the specific case of asymptotic tensor rank, noting that it has Zariski-closed sublevel sets (for fixed formats) by~\cref{th:intro-2}.
\begin{proof}[Proof of \cref{cor:closed}]
    Let $\mathcal{R} = \{ \asymprank(T) : T \in \CC^{d_1} \otimes \cdots \otimes \CC^{d_k}, d \in \ZZ^k_{\geq1} \}$.
    Let $r_1, r_2, \ldots$ be a converging sequence in $\mathcal{R}$ with limit $r \in \RR$.
    Then the sequence is bounded, say $r_i \leq m$ for all~$i$.
    Let $T_i$ be tensors such that $\asymprank(T_i) = r_i$.
    We may assume that the $T_i$ are concise. Then for every $i$, for some $d_1, \ldots, d_k\leq m$ we have $T_i \in \CC^{d_1} \otimes \cdots \otimes \CC^{d_k}$. Thus our sequence $r_1, r_2, \ldots$ is in $\cup_{d_i \leq m} \{ \asymprank(T) : T \in \CC^{d_1} \otimes \cdots \otimes \CC^{d_k}\}$ which is a closed set. We conclude that $r \in \cup_{d_i \leq m} \{ \asymprank(T) : T \in \CC^{d_1} \otimes \cdots \otimes \CC^{d_k}\} \subseteq \mathcal{R}$. %
\end{proof}

The same argument extends from asymptotic tensor rank to every point in Strassen's asymptotic spectrum, by combining \cref{cor:eucl-cl-general} with the growth ingredient used in the proof of \cref{th:spec-well-ord}.
\begin{corollary}\label{cor:spec-eucl-closed}
Let $k \in \ZZ_{\geq 2}$ and $F \in \Delta(\CC, k)$. Then
$\{F(T) : T \in \CC^{d_1} \otimes \cdots \otimes \CC^{d_k},\, d \in \ZZ_{\geq1}^k\}$ is Euclidean-closed in~$\RR$.
\end{corollary}
\begin{proof}
The argument mirrors the proof of \cref{th:spec-well-ord}, by induction on~$k$. The base case $k=2$ is matrix rank, for which the claim is trivial. Suppose $k > 2$ and let $r_1, r_2, \ldots$ be a converging sequence in the set, with limit~$r$ and bound $r_i \leq m$ for all~$i$.

Suppose first that $F(\langle 2\rangle_{i,j}) > 1$ for every $i\neq j \in [k]$. Then, as in the proof of \cref{th:spec-well-ord}, $F(T)$ is at least some power of $\max_i \tensorrank_i(T)$ for every $k$-tensor $T$. Picking concise representatives $T_i$ with $F(T_i) = r_i$, the bound $r_i \leq m$ forces all format dimensions to be at most some $M \in \NN$. Hence the sequence lies in the finite union $\bigcup_{d_i \leq M} \{F(T) : T \in \CC^{d_1} \otimes \cdots \otimes \CC^{d_k}\}$. Each set in this union is Euclidean-closed by \cref{cor:eucl-cl-general} (which applies since $F$ has Zariski-closed sublevel sets, by \cref{cor:general-lowerset-closed}), so the union is Euclidean-closed and the limit $r$ lies in it.

Suppose instead that $F(\langle 2\rangle_{i,j}) = 1$ for some $i\neq j \in [k]$. By relabeling legs we may take $i = k-1$ and $j = k$. By \cref{lem:spec-descend}, $F$ takes the same set of values as some $F' \in \Delta(\CC, k-1)$, and the claim follows by the induction hypothesis.
\end{proof}

\subsection{Geometric characterization of discreteness from below}

In this section our base field is the complex numbers. For a complex vector space $V$ we consider any function $F\colon V \to \RR$ whose sublevel sets are Zariski-closed (e.g., asymptotic tensor rank, \cref{cor:general-lowerset-closed}). Letting $\mathcal{\func}$ be the set $\{ \func(T) : T \in V \}$, it follows from topological Noetherianity of $V$ that $\mathcal{\func}$ is well-ordered (cf.~the proof of \cref{cor:asymprank-well-order}). In this section we give a geometric characterization for such a set $\mathcal{F}$ to be also discrete from below (that is, for every nondecreasing sequence in it to stabilize).

\begin{theorem}\label{th:geom-charac}
    Let $V$ be a complex vector space, and let $\func\colon V \to \RR$ be Zariski-lower-semicontinuous.
    We write~$V_{\leq r} = \{ T \in V : \func(T) \leq r \}$ and $V_{=r} = \{ T \in V : \func(T) = r \}$ for~$r \in \RR$.
    The following are equivalent.
    \begin{enumerate}[\upshape(a)]
        \item $\{ \func(T) : T \in V \}$ is discrete from below. %
        \item For every $r \in \RR$, $V_{=r}$ is Zariski-open in $V_{\leq r}$.
    \end{enumerate}
\end{theorem}
Note in \cref{th:geom-charac} that $V_{=r}$ may be empty (namely when $F$ does not attain $r$), in which case $V_{=r}$ is trivially Zariski-open.
\begin{proof}
    (a) $\Rightarrow$ (b). 
    Let $r \in \RR$.
    Since $\mathcal{\func}$ is discrete from below, there exists an $r'\in \RR$ such that $r' < r$ and for all $s < r$ with $s \in \mathcal{\func}$ we have
    $s < r'$ .
    Then
    \begin{equation}
        V_{=r} = V_{>r'} \cap V_{\leq r}
    \end{equation}
    where $V_{>r'} = \{ T \in V : \func(T) > r' \}$ is Zariski-open in $V$ by the assumption on~$F$.
    Therefore $V_{=r}$ is open in~$V_{\leq r}$.

    (b) $\Rightarrow$ (a). 
    It suffices to show that if $r_1 \leq r_2 \leq \dotsc$ is a sequence in $\mathcal{\func}$, then there exists $i_0 \geq 1$ such that $r_i = r_{i+1}$ for all $i \geq i_0$.
    We observe that any such sequence is bounded above by~\cref{thm:irreducible-variety-rank-maximizers-dense-general}, hence $r := \lim_{i\to\infty} r_i$ exists, and~$r \in \mathcal{\func}$, because $\mathcal{\func}$ is closed (\cref{cor:eucl-cl-general}).
    We now argue by contraposition.
    If $V_{=r}$ is open, $V_{<r}$ is closed.
    If $r_1 < r_2 < \cdots$ is a strictly increasing sequence of values of~$F$ converging to $r$, then
    \begin{equation*}
        V_{<r} = \cup_{r_i < r} V_{\leq r_i}
    \end{equation*}
    expresses an affine variety as a countably infinite union of proper closed subvarieties.
    This is impossible over $\CC$, as we will now prove. (Note, however, that for instance over the algebraic closure of the rationals, which is countable, this is possible.) 
    Indeed if $V_{<r} = X_1 \cup \dotsc \cup X_\ell$ is its decomposition into irreducible components, then there is some component $X_j$ such that for every $r_i < r$, $V_{\leq r_i} \cap X_j$ is a proper closed subset of $X_j$, for otherwise every component would be of the form~$V_{\leq r_{i_j}}$ for some~$i_j$, and one would have~$V_{<r} = \cup_{j=1}^\ell V_{\leq r_{i_j}}$, contradicting the assumption that we have an infinite sequence of strictly increasing~$V_{\leq r_i}$'s. Therefore $V_{\leq r_i} \cap X_j$ has non-empty Zariski-open complement (as $X_j$ is irreducible), which is also dense.
    Therefore $\cup_{r_i<r} V_{\leq r_i} \cap X_j$  has dense complement in $X_j$ by \autoref{lem:baire-property}, and is in particular non-empty.
    As a result we obtain that $\cup_{r_i<r} V_{\leq r_i}$ cannot contain all of $X_j$.
\end{proof}

\begin{corollary}\label{cor:asymprank-disc-equiv}
    Let $V = \CC^{d_1} \otimes \cdots \otimes \CC^{d_k}$, and let $V_{\leq r}$ be the set of all tensors in $V$ with asymptotic rank at most $r$. %
    Let $V_{=r}$ be the set of all tensors in $V$ with asymptotic rank equal to $r$.
    The following are equivalent.
    \begin{enumerate}[\upshape(a)]
        \item $\{ \asymprank(T) : T \in V \}$ is discrete.
        \item For every $r \in \RR$, $V_{=r}$ is Zariski-open in $V_{\leq r}$.
    \end{enumerate}
\end{corollary}
\begin{proof}
This follows from~\cref{th:geom-charac} and~\cref{cor:general-lowerset-closed}.
\end{proof}

\section{Discussion and open problems}
We discuss several natural directions and open problems in the context of our results, Strassen's asymptotic rank conjecture and the matrix multiplication exponent.

\begin{itemize}
\item \textbf{Discreteness from below.} We proved discreteness from above for the asymptotic rank (and a large class of other parameters) over (for instance) the complex numbers. One of the main problems that we leave open is whether this parameter is also discrete from below. Indeed, Strassen's asymptotic rank conjecture would imply this. 

As an intermediate result, we have shown that any converging sequence of asymptotic ranks has a limit which is also an asymptotic rank (\cref{cor:closed}). What remains to be shown is that any such sequence is eventually constant; we proved an equivalent characterization in \cref{cor:asymprank-disc-equiv}. 

Notions of discreteness for asymptotic parameters have been studied more broadly, in particular in the context of the Shannon capacity of graphs \cite{MR0089131} and the related asymptotic spectrum of graphs \cite{MR4039606, MR4357434}; interestingly, in that setting, the asymptotic parameter of interest (Shannon capacity) is not discrete (neither from above or from below), which leads to a graph limit approach to determining Shannon capacity~\cite{deboer2024asymptoticspectrumdistancegraph}.

\item \textbf{Geometric properties; irreducibility.}
Given (\cref{th:intro-2}) that the sublevel sets $\{T \in \FF^{d_1} \otimes \cdots \otimes \FF^{d_k} : \asymprank(T) \leq r\}$ of the asymptotic tensor rank are Zariski-closed, it is natural to ask about the geometric properties of these sets. For instance, we may ask if they are irreducible (i.e., cannot be written as the union of two proper Zariski-closed subsets) whenever $\FF$ is algebraically closed. 

Indeed, irreducibility is true for $k=2$ (matrices), since then the asymptotic rank coincides with matrix rank, and the set of matrices of at most a given rank is an irreducible variety. 
For $k \geq 3$, irreducibility is open. It would imply both that the set of achievable values of the asymptotic rank is discrete, and a weak form of Strassen's asymptotic rank conjecture:
there exist at most $d_1 \dotsb d_k + 1$ asymptotic ranks in format $d_1 \times \dotsb \times d_k$.
This follows from a dimension argument (topological dimension is the maximal length of a decreasing chain of irreducible subvarieties, and the dimension of $\FF^{d_1 \times \dotsb \times d_k}$ is $d_1 \dotsb d_k$).

To see that this is a weak form of Strassen's conjecture, note that Strassen's conjecture is equivalent to the statement that there are exactly $d + 1$ distinct asymptotic ranks in~$(\FF^d)^{\otimes k}$ (namely~$\{0, 1, \dotsc, d\}$).
Note that Strassen's conjecture also implies that the sets $\{T \in \FF^{d_1} \otimes \cdots \otimes \FF^{d_k} : \asymprank(T) \leq r\}$ are irreducible: $r$ is integer, and $\FF^{d_1 \times r} \times \dotsb \times \FF^{d_k \times r} \times \FF^{r \times \dotsb \times r}$ is irreducible and admits a surjective polynomial map onto the set of tensors with all flattening ranks at most $r$.

\item \textbf{Computation and structure of asymptotic ranks.} Our results provide new tools to concretely understand the asymptotic rank of families of tensors and their relation to the matrix multiplication exponent $\omega$, in the spirit of \cite{kaski2024universalsequencetensorsasymptotic}. For instance (by \cref{th:intro-2}), truth of the asymptotic rank conjecture on any subset of tensors by our result extends to the Zariski-closure. 

Another natural direction where our result plays a role is in the task of understanding the relation between asymptotic ranks of explicit (families of) tensors and the matrix multiplication exponent. For example, it is well-known that if the asymptotic rank of the small Coppersmith--Winograd tensor $\mathrm{cw}_2$ equals 3, then $\omega=2$. For which tensors can we prove this property? 

It follows from \cref{th:intro-2} and known classifications that $\mathrm{cw}_2$ has asymptotic rank at most the generic asymptotic rank of $3 \times 3 \times 3$ tensors with hyperdeterminant $0$ (which is a codimension $1$ variety). It also has at most the asymptotic rank of a generic tensor with support
    $\{(0,0,0), (1,1,1), (2,2,2), (0,1,2), (1,2,0), (2,0,1)\}$.
\end{itemize}

\paragraph{Acknowledgements.}
JZ was supported by NWO Veni grant VI.Veni.212.284, Vidi grant VI.Vidi.243.195 and ERC Starting Grant 101220349 SPECTRA. KH was supported by NWO M-1 grant OCENW.M.24.147.
MC and HN thank the European Research Council (ERC Grant Agreement No.~818761), VILLUM FONDEN via the QMATH Centre of Excellence (Grant No.~10059) and the Novo Nordisk Foundation (Grant NNF20OC0059939 ``Quantum for Life'') for financial support. Part of this work was completed while MC was Turing Chair for Quantum Software, associated to the QuSoft research center in Amsterdam, acknowledging financial support by the Dutch National Growth Fund (NGF), as part of the Quantum Delta NL visitor programme.
PV was supported by the Ministry of Culture and Innovation of Hungary from the National Research, Development and Innovation Fund, financed under the FK~146643 funding scheme, by the János Bolyai Research Scholarship of the Hungarian Academy of Sciences, and by the Ministry of Culture and Innovation and the National Research, Development and Innovation Office within the Quantum Information National Laboratory of Hungary (Grant No.~2022-2.1.1-NL-2022-00004).

\paragraph{Formalisation.} 
The results of this paper are formalized in Lean:
\url{https://github.com/spectra-research/asymptotic-tensor-rank-semicontinuity.git}

\bibliographystyle{alphaurl}
\bibliography{refs}

\end{document}